%% file: clar.tex
\documentclass{article}
\usepackage[margin=27mm,nohead]{geometry}

\usepackage{color}

\usepackage[cmex10]{amsmath}
\usepackage{amsfonts,amsthm,graphicx}

\input{pszkod.ltx}

\newcommand{\Rset}{\mathbb{R}}

\usepackage[normalem]{ulem} 

\newtheorem{theorem}{Theorem}
\newtheorem{problem}{Problem}
\newtheorem{lemma}{Lemma}
\newtheorem{defn}{Definition}
\newtheorem{claim}{Claim}
\newtheorem{cor}{Corollary}

\begin{document}
\title{The complexity of the Clar number problem and an FPT algorithm} 
\author{Erika R. B\'erczi-Kov\'acs\thanks{Department of Operations Research, E\"otv\"os University, P\'azm\'any P\'eter s\'et\'any 1/C, Budapest, Hungary, H-1117.  Supported by  the Hungarian Scientific Research Fund (OTKA, grant number K109240).} 
\and 
Attila Bern\'ath\thanks{MTA-ELTE Egerv\'ary Research Group,
Department of Operations Research, E\"otv\"os University, P\'azm\'any P\'eter s\'et\'any 1/C, Budapest, Hungary, H-1117.  Supported by  the Hungarian Scientific Research Fund (OTKA, grant number K109240).}
}

\maketitle

\begin{abstract}
  The Clar number of a (hydro)carbon molecule, introduced by Clar
    [E.~Clar, \emph{The aromatic sextet}, (1972).], is the maximum
    number of mutually disjoint resonant hexagons in the molecule.
    Calculating the Clar number can be formulated as an optimization
    problem on 2-connected planar graphs.  Namely, it is the maximum
    number of mutually disjoint even faces a perfect matching can
    simultaneously alternate on.  It was proved by Abeledo and
    Atkinson [H.~G. Abeledo and G.~W. Atkinson, \emph{Unimodularity of
        the clar number problem}, Linear algebra and its applications
      \textbf{420} (2007), no.~2, 441--448] that the Clar number can
    be computed in polynomial time if the plane graph has even faces
    only.  We prove that calculating the Clar number in general
  2-connected plane graphs is NP-hard.  We also prove NP-hardness of
    the maximum independent set problem for 2-connected plane graphs
    with odd faces only, which may be of independent interest.
  Finally, we give an FPT algorithm that determines the
  Clar number of a given 2-connected plane graph.  The parameter of
  the algorithm is the length of the shortest odd join in the planar
  dual graph.  For fullerenes this is not yet a polynomial algorithm,
  but for certain carbon nanotubes it gives an efficient algorithm.
\end{abstract}

\begin{quote}
\end{quote}
\vspace{5mm}

\section{Introduction}

Our research is motivated by problems in chemical graph theory. Some
molecules, for example polycyclic aromatic hydrocarbon (PAH) molecules,
benzenoid hydrocarbon molecules, or fullerene molecules can be
represented as a 2-connected plane graph. In this representation only carbon atoms are depicted, while hydrogen atoms are omitted. 
Several of the chemical properties of these molecules (e.g. chemical stability) are closely related to the parameters of the underlying graph. In this note we will be concerned with one of these parameters, the \textbf{Clar number} that we will define later.

A subclass of PAHs, the benzenoid PAHs have the special property that every
face in this plane graph is a hexagon, in particular, every face has an even
number of nodes. One can see that such a 2-connected plane graph is
also bipartite. Abeledo and Atkinson proved that the Clar number of a
2-connected bipartite plane graph can be computed in polynomial
time. On the other hand we show in this note that determining the Clar
number of a general 2-connected plane graph is NP-hard.

Fullerenes are carbon molecules with a hollow cage-like structure.
The first fullerene molecule to be discovered, and the family's
namesake, buckminsterfullerene ($C_{60}$), was prepared in 1985 by
Richard Smalley, Robert Curl, James Heath, Sean O'Brien, and Harold
Kroto at Rice University \cite{kroto1985c}. The graph representing a fullerene molecule
contains exactly 12 pentagon faces, and the rest of the faces are
hexagons (the number of hexagons can be arbitrarily large). 
For the Clar number of fullerenes Ye and Zhang gave an upper bound of $\lfloor \frac{n-12}{6}\rfloor$ \cite{zhang2007upper}.
Later they characterized the fullerenes achieving this bound \cite{ye2009extremal}.

One of our
motivations was to determine the Clar number of fullerene molecules in
polynomial time. 
We show that
determining the Clar number of a general 2-connected plane graph is
NP-hard, if the number of odd faces is not bounded in the planar
embedding\footnote{Fullerene graphs are quite special 2-connected plane graphs: for example they are 3-connected and  3-regular, too. We believe that among these properties the constant bound 12 on the number of odd faces is the most important, if we are concerned with their Clar number.}. 
We present an algorithm that determines the Clar number of a
2-connected plane graph, and has good running time, provided that the
odd faces are ``not too far from each other''. More precisely, our
algorithm is fixed parameter tractable (FPT) where the parameter is
the length of the shortest odd join in the planar dual graph. In
Section \ref{sec:alg} we explain that for a subclass of fullerenes
(called carbon nanotubes) our algorithm efficiently computes the Clar
number.


Let $G=(V,E)$ denote a $2$-connected planar graph which has a perfect
matching. For a planar embedding of $G$ and a perfect matching of $G$
let $F_M$ denote the set of those faces which alternate with respect
to $M$. Note that faces in $F_M$ are even. A pairwise vertex disjoint
subset of $F_M$ is a \textbf{Clar set with respect to $M$}. A subset
$C$ of the faces is a \textbf{Clar set} if there exists a perfect
matching $M$ for which $C$ is a Clar set with respect to $M$.  Note
that a set of pairwise vertex disjoint even faces is a Clar set if and
only if deleting all (the nodes of) these even faces the remaining
graph still has a perfect matching. The \textbf{Clar number} of $G$,
denoted by $Cl(G)$ is the maximum size of a Clar set. For sake of
simplicity we allow the unbounded face in a Clar set as well, but
there are no difficulties if we want to exclude it. By a \textbf{
  plane graph} we mean a planar graph with a fixed planar embedding.
For further graph theoretic definitions we refer the reader to \cite{frank2011connections}.

The Clar number was defined by Clar in \cite{clar}. It was proved by
Abeledo and Atkinson \cite{abeledo_minmax} that the Clar number can be
computed in polynomial time if $G$ is bipartite. Note that a plane
graph is bipartite if and only if all its faces are even. 
Our first result is the following theorem.


\begin{theorem}\label{thm:nph}
It is NP-hard to calculate the Clar number of a 2-connected planar
graph (given with a fixed planar embedding).
\end{theorem}

{The proof of the above theorem will be detailed in Section
  \ref{sec:compl}, based on a reduction from a special case of the
  independent set problem.} In Section \ref{sec:alg} we present an FPT
algorithm that determines the Clar number of a 2-connected plane
graph.


\section{Hardness of the Clar number problem}
\label{sec:compl}

In this section we prove Theorem \ref{thm:nph}. Our reduction will be
based on a special case of the Independent Set Problem. Let us start
with defining this problem.

\begin{defn}
Given a graph $G=(V,E)$, a subset $U\subseteq V$ is said to be
\textbf{independent} if there is no edge of $G$ between two nodes of $U$.
Let $\alpha(G)$ denote the maximum size of an
independent set in $G$.
\end{defn}

%
%
%
%

\begin{problem}\label{prob:2connplcubic}
Given a 2-connected planar cubic graph $G$ and a positive integer $K$, does $G$
contain an independent set of size $K$?
\end{problem}

\begin{theorem}[Mohar, Theorem 4.1 in \cite{mohar2001face}]\label{thm:2connplcubic}
Problem \ref{prob:2connplcubic} is NP-complete.
\end{theorem}

\begin{problem}\label{prob:2connplodd}
Given a 2-connected plane graph $G$ with odd faces only, and a positive integer $K$, does $G$
contain an independent set of size $K$?
\end{problem}

\begin{lemma}\label{indep}
Problem \ref{prob:2connplodd} is NP-hard.
\end{lemma}
\begin{proof}
{According to Theorem \ref{thm:2connplcubic}, the independent set problem is also NP-complete for 2-connected planar graphs.}
Let $G=(V,E)$ denote an instance of
this problem{, and let us fix a planar embedding of $G$}.  
If $G$ has an even face $F$, let $G_F$ denote the
planar graph obtained from $G$ by the following operation. We add
three vertices $a, b, c$ inside $F$ and edges $ab, bc, ca, au, bu, bv$
where $u$ and $v$ form an edge of $F$ (see Figure \ref{fig_odd}).
\begin{figure}[!t]
\begin{center}
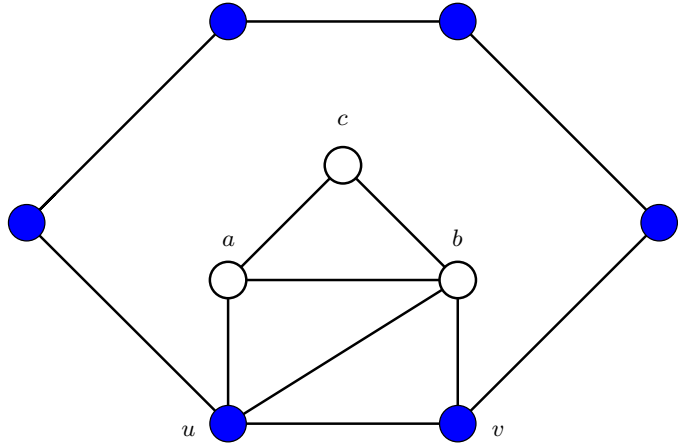
\caption{Eliminating even faces.}
\label{fig_odd}
\end{center}
\end{figure}

\begin{claim}
$\alpha(G_F)=1+\alpha(G)$.
\end{claim}
\begin{proof}
First, for an independent set $I$ of $G$, clearly $I\cup c$ is independent in $G_F$ and hence $\alpha(G_F)\geq 1+\alpha(G)$. 
Second, an independent set $I_F$ in $G_F$ can contain at most one vertex from the set $\{a, b, c\}$. Since $I_F\setminus \{a, b, c\}$ is independent in $G$ we get that $\alpha(G)\geq \alpha(G_F)-1$. 
\end{proof}
Note that the number of even faces of $G_F$ is one less than that of $G${, and $G_F$ is also 2-connected}. Let $\mathbb{F}$ denote the set of even faces of $G$. By consecutively applying the above operation on every member of $\mathbb{F}$ we get another graph $G_{\mathbb{F}}$ for which $\alpha(G_{\mathbb{F}})=\alpha(G)+|\mathbb{F}|$ and which has odd faces only. 
\end{proof}

We are now ready to prove the hardness of the Clar number problem.

\begin{proof}[Proof of Theorem \ref{thm:nph}]
We prove the theorem by reducing Problem \ref{prob:2connplodd} to the
Clar number problem.  Let $G=(V,E)$ denote an instance of this
problem. We construct graph $G'$ the following way: for every edge of
$G$ we add two vertices to $G'$. Let $uv\in E$ be an edge of $G$ and
let $F_1$ and $F_2$ denote the faces $uv$ is incident to. We add
vertices $x_{uv,F_1}$ and $x_{uv,F_2}$ to $G'$ along with the edge
$x_{uv,F_1}x_{uv,F_2}$. If edges $uv$ and $vw$ are neighbouring edges
on a face $F$, then we add edge $x_{uv,F}x_{vw,F}$ to $G'$. It is easy
to see that $G'$ is planar (see Figure \ref{fig_dual}). Informally,
$G'$ is obtained from the planar dual graph $G^*$ of $G$ by ``blowing
a circuit'' into each vertex of $G^*$. Every face of $G'$ either
corresponds to a face of $G$, or a vertex of $G$, and since $G$ has odd
faces only, all the even faces of $G'$ are the ones corresponding
to vertices of $G$. Note that $G'$ trivially has a perfect matching
$M$ consisting of the edges of the form $x_{uv,F_1}x_{uv,F_2},$ for every
$uv \in E$. Since $M$ is alternating on every even face of $G'$,
corresponding to a vertex of $G$, that is, on every even face of $G'$,
for this graph the Clar number equals the maximum size of a Clar set
with respect to $M$. The Clar sets of $G'$ and the independent sets of
$G$ have a one to one correspondence, proving the theorem.
\begin{figure}[!t]
\begin{center}
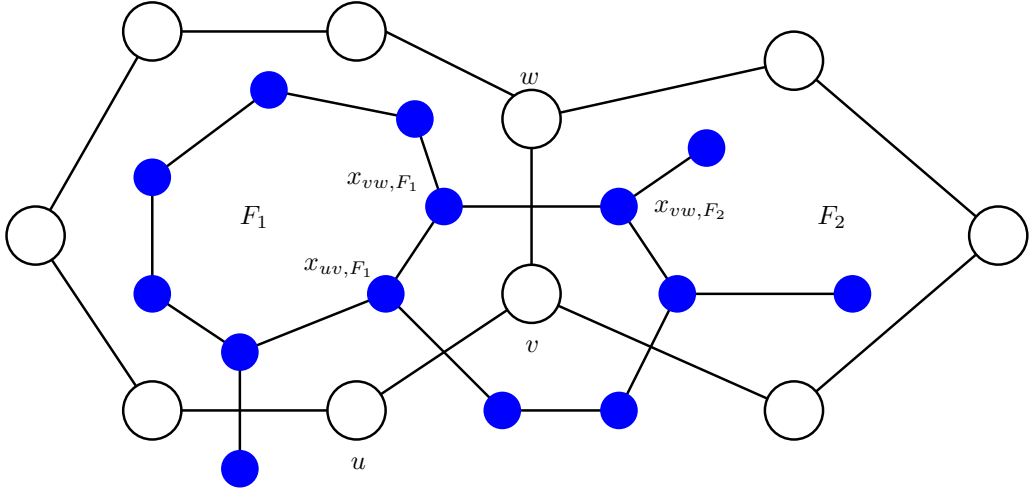
\caption{Reduction of the Independent Set Problem to the Clar number problem}
\label{fig_dual}
\end{center}
\end{figure}
\end{proof}

\begin{cor}
It is also NP-hard to find a maximum cardinality Clar set with respect to a fixed perfect matching. 
\end{cor}

\section{An FPT algorithm for determining the Clar number}
\label{sec:alg}

In this section we present an algorithm that determines the Clar
number of a 2-connected plane graph, and has a good running time,
unless the odd faces are ``far from each other'' in the planar
representation.  The idea is the following. Consider a 2-connected
plane graph that has only 2 odd faces in its (fixed) planar
representation, and take a shortest path (in the planar dual graph)
between these odd faces. An optimal Clar set might use some of the
even faces that lie on this shortest path. Our algorithm takes an
arbitrary subset of even faces along this shortest path and tries to
extend this subset into a Clar set. This is repeated for every
possible subset of even faces along the shortest path. We will
generalize this for plane graphs having more than 2 odd faces
below. First we need a definition and a theorem.



\begin{defn}
Given a graph $G=(V, E)$ and a subset $T\subseteq V$ of even size, a
\textbf{$T$-join} is a subset of edges $J\subseteq E$ so that the number of
edges of $J$ incident to a node $v\in V$ is odd if and only if $v\in T$. An \textbf{odd-join} of $G$ is a $T$-join where $T$ is the set of nodes having odd degree in $G$.
\end{defn}

\begin{theorem}[See e.g. \cite{lexbook}, Chapter 29 ]
Given a graph $G=(V, E)$, a subset $T\subseteq V$ of even size, and edge-lengths $c:E\to \Rset_+$, a shortest $T$-join can be found in polynomial time.
\end{theorem}


Given a 2-connected plane graph $G=(V,E)$, let $G^*=(V^*, E^*)$ denote
its planar dual. Let $T\subseteq V^*$ be the set of odd-degree nodes
of $G^*$. Let $J^*\subseteq E^*$ be a shortest $T$-join in $G^*$,
where each edge of $G^*$ has length 1 (we refer the reader to \cite[Chapter 29]{lexbook} for preliminaries on $T$-joins). We give an algorithm
determining the Clar number of $G$ that runs in $O(3^{|J^*|}p(|V|))$ for some polynomial $p$.

\newcommand{\Feven}{\ensuremath{F_{even}}}

Let $J\subseteq E$ be the set of edges corresponding to
  $J^*$. Let $\Feven$ be the set of even faces of $G$ (that is,
  $|\Feven|= |V^*|-|T|$) and let $F_J\subseteq \Feven $ be the set of
  even faces that have some edge of $J$ in their boundary.
 Let $G'=(V+U, E - J + J')$ be the 2-connected bipartite plane
  graph that is obtained from $G$ by subdividing each edge of $J$ with
  a new node (the set of these subdivision nodes is $U$, the set of
  subdivided edges is $J'$: observe that $|U|=|J|$ and $|J'|=2|J|$). Note that $G'$ is indeed bipartite, since every face is even in its planar embedding. 

Let $K'$ be the node-edge incidence matrix of $G'$, and $R'$ be the
node-face incidence matrix of $G'$. Let $R$ be obtained from $R'$ by
deleting the columns corresponding to odd faces of $G$, and let $K$ be
obtained from $K'$ by deleting the columns corresponding to $J'$. Since the matrix $[R',K']$ is unimodular by Theorem  3.5 of \cite{abeledo2007unimodularity} we get the following claim.

\begin{claim} \label{cl:unimod}
The matrix $[R,K]$ is unimodular.
\end{claim}

After these preliminaries we present the pseudocode of our algorithm
that calculates the Clar number of a 2-connected plane graph. The
basic idea of the algorithm is the following. Determining the Clar
number of $G$ means that we want to choose pairwise node disjoint even
faces and edges, so that every node is contained in exactly one of the
chosen objects, and we want to maximize the number of faces
chosen. Given a feasible solution consisting of a set $F_1$ of even
faces of $G$ and a set $E_1$ of edges of $G$, let $F_1'$ be the set of
even faces of $G'$ corresponding to faces in $F_1$, and similarly let
$E_1'$ be set of edges of $G'$ corresponding to edges in $E_1$ (if an
edge $e\in J$ is in $E_1$ then we add both edges obtained from the
subdivision of $e$ into $E_1'$). Every subdivision node (that is, node
in $U$) is then incident to either $0, 1$ or $2$ objects in $F_1'\cup
E_1'$. If someone tells us these $0,1,2$ values for every $u\in U$
then we can reconstruct $E_1$ and $F_1$ using these numbers, see Lemma
\ref{lem:intpoly} below. Therefore what we do is that we try every
possible vector in $\{0,1,2\}^U$ to find the one giving the best
solution.



\begin{pszkod}{Algorithm Clar\_Number}

\item[] INPUT: a 2-connected plane graph $G=(V,E)$

\item[] OUTPUT: the Clar number of $G$

\item Find a shortest odd-join in $G^*$,
where each edge of $G^*$ has length 1
($G^*$ is the planar dual of $G$, and we will use more  notations that were introduced above in this section).




\item For every vector $b_U\in \{0,1,2\}^U$ \label{st:iter}

\tab

  \item Let $b=\left(
    \begin{array}{c}
    1_V\\
    b_U
    \end{array}
    \right)\in \{0,1,2\}^{V+U}$. 

  \item For every $e\in J'$
    \tab

      \item Let $z_e = 1$  if $e$ is incident with a node $u\in U$ with $b_U(u) = 2$, and let $z_e = 0$ otherwise.

    \untab

  \item \label{st:LP} Take the integer optimum of the LP Problem
    \eqref{eq:LP1}-\eqref{eq:LPut} (see Lemma \ref{lem:intpoly}).
  \begin{eqnarray}
    \max\{1y: y\in \Rset_+^{\Feven}, x\in \Rset_+^{E-J+J'},\label{eq:LP1}\\
    Ry + K' x = b, 
    x_e=z_e\mbox{ for every  }e\in J'.\label{eq:LPut}
  \}
  \end{eqnarray}

\untab

\item Output the best of the candidates obtained in Step \ref{st:LP}.

\end{pszkod}

\begin{lemma}\label{lem:intpoly}
The LP Problem \eqref{eq:LP1}-\eqref{eq:LPut} has an integer optimum.
\end{lemma}
\begin{proof}
After eliminating the variables $x_e$ for $e\in J'$ we obtain an LP
Problem of the form $\max\{1y: y\in \Rset_+^{\Feven}, x\in
\Rset_+^{E-J}, Ry + K x = b'\}$. The polyhedron in this problem is
integral by Claim \ref{cl:unimod}.
\end{proof}

Note that the  LP Problem \eqref{eq:LP1}-\eqref{eq:LPut} will not necessarily be feasible for every choice of $b_U$. We could be more careful in choosing only those vectors in Step \ref{st:iter} of the algorithm that make the LP feasible. However the algorithm is easier described this way. The running time is clearly $O(3^{|J^*|}p(|V|))$ for some polynomial $p$.

Carbon nanotubes are fullerenes with a cylindrical nanostructure, with two 'half-fullerene' caps on both ends. Six pentagonal faces are in both caps, forming three short pairs in the odd join. So for this class of fullerenes the parameter of our FPT algorithm is relatively small, giving an efficient method to determine the Clar number.

\section{Open questions}
We have proved the NP-hardness of the Clar number problem for general plane
graph $G$.  The problem is motivated by the problem of determining the Clar number of fullerene graphs, when  $G$ has exactly
twelve pentagonal faces and every other face is a hexagon. This problem is
however left open, since our NP-hardness
reduction involves creating a lot of odd faces. An FPT algorithm with the number of odd faces as parameter would yield a polynomial time algorithm for all fullerenes.

Another line of research would be to show that determining the Clar
number is NP-hard even for some restricted class of 2-connected
graphs, too. If we were able to specialize the Independent Set problem
further to $3$-regular plane graphs with odd faces, then our
techniques would yield that the Clar number is NP-hard for graphs
with only hexagonal even faces.

\bibliographystyle{amsplain}
\bibliography{bclar}

 



\end{document}

%% file: pszkod.ltx
\makeatletter

\newcounter{Kulso}

\newenvironment{pszkod}[2][1]
{
  \setcounter{Kulso}{#1}
  \sffamily
  \vspace{\topsep}              
  \noindent
  {#2}
  \begin{list}{\rmfamily\arabic{Kulso}.\@arabic\c@enumiv.}
    {\usecounter{enumiv}%
     \renewcommand\p@enumiv{\arabic{Kulso}.}
     \renewcommand\theenumiv{\@arabic\c@enumiv}
     \setlength{\parsep}{0pt}
     \setlength{\itemsep}{0pt}
     \setlength{\topsep}{0pt}
    }
  \psz@cimke{begin}
}
{
  \psz@cimke{end}
  \end{list}
  \vspace{\topsep}
}

\newlength{\tabhossz}
\newcommand{\tab}{\par
\advance\@totalleftmargin \tabhossz
\advance\linewidth -\tabhossz
\parshape \@ne \@totalleftmargin \linewidth
\addtolength{\labelsep}{\tabhossz}%
}
\newcommand{\untab}{\par
\advance\@totalleftmargin -\tabhossz   
\advance\linewidth \tabhossz
\parshape \@ne \@totalleftmargin \linewidth
\addtolength{\labelsep}{-\tabhossz}%
}

\newcommand{\psz@cimke}[1]{\item[{\makebox[\labelwidth][l]{#1}}]}

\makeatother

%% file: odd.pdf_t
\begin{picture}(0,0)%
\includegraphics{odd.pdf}%
\end{picture}%
\setlength{\unitlength}{4144sp}%
\begingroup\makeatletter\ifx\SetFigFont\undefined%
\gdef\SetFigFont#1#2#3#4#5{%
  \reset@font\fontsize{#1}{#2pt}%
  \fontfamily{#3}\fontseries{#4}\fontshape{#5}%
  \selectfont}%
\fi\endgroup%
\begin{picture}(4018,2646)(966,-2547)
\put(2941,-638){\makebox(0,0)[lb]{\smash{{\SetFigFont{9}{10.8}{\rmdefault}{\mddefault}{\updefault}{\color[rgb]{0,0,0}$c$}%
}}}}
\put(3629,-1360){\makebox(0,0)[lb]{\smash{{\SetFigFont{9}{10.8}{\rmdefault}{\mddefault}{\updefault}{\color[rgb]{0,0,0}$b$}%
}}}}
\put(2252,-1360){\makebox(0,0)[lb]{\smash{{\SetFigFont{9}{10.8}{\rmdefault}{\mddefault}{\updefault}{\color[rgb]{0,0,0}$a$}%
}}}}
\put(2008,-2498){\makebox(0,0)[lb]{\smash{{\SetFigFont{9}{10.8}{\rmdefault}{\mddefault}{\updefault}{\color[rgb]{0,0,0}$u$}%
}}}}
\put(3866,-2498){\makebox(0,0)[lb]{\smash{{\SetFigFont{9}{10.8}{\rmdefault}{\mddefault}{\updefault}{\color[rgb]{0,0,0}$v$}%
}}}}
\end{picture}%

%% file: dual.pdf_t
\begin{picture}(0,0)%
\includegraphics{dual.pdf}%
\end{picture}%
\setlength{\unitlength}{4144sp}%
\begingroup\makeatletter\ifx\SetFigFont\undefined%
\gdef\SetFigFont#1#2#3#4#5{%
  \reset@font\fontsize{#1}{#2pt}%
  \fontfamily{#3}\fontseries{#4}\fontshape{#5}%
  \selectfont}%
\fi\endgroup%
\begin{picture}(6137,2926)(1101,-3413)
\put(4187,-990){\makebox(0,0)[lb]{\smash{{\SetFigFont{10}{12.0}{\rmdefault}{\mddefault}{\updefault}{\color[rgb]{0,0,0}$w$}%
}}}}
\put(3175,-3293){\makebox(0,0)[lb]{\smash{{\SetFigFont{10}{12.0}{\rmdefault}{\mddefault}{\updefault}{\color[rgb]{0,0,0}$u$}%
}}}}
\put(4990,-1757){\makebox(0,0)[lb]{\smash{{\SetFigFont{10}{12.0}{\rmdefault}{\mddefault}{\updefault}{\color[rgb]{0,0,0}$x_{vw,F_2}$}%
}}}}
\put(4221,-2595){\makebox(0,0)[lb]{\smash{{\SetFigFont{10}{12.0}{\rmdefault}{\mddefault}{\updefault}{\color[rgb]{0,0,0}$v$}%
}}}}
\put(2895,-2106){\makebox(0,0)[lb]{\smash{{\SetFigFont{10}{12.0}{\rmdefault}{\mddefault}{\updefault}{\color[rgb]{0,0,0}$x_{uv,F_1}$}%
}}}}
\put(5967,-1827){\makebox(0,0)[lb]{\smash{{\SetFigFont{10}{12.0}{\rmdefault}{\mddefault}{\updefault}{\color[rgb]{0,0,0}$F_2$}%
}}}}
\put(2512,-1827){\makebox(0,0)[lb]{\smash{{\SetFigFont{10}{12.0}{\rmdefault}{\mddefault}{\updefault}{\color[rgb]{0,0,0}$F_1$}%
}}}}
\put(3151,-1591){\makebox(0,0)[lb]{\smash{{\SetFigFont{10}{12.0}{\rmdefault}{\mddefault}{\updefault}{\color[rgb]{0,0,0}$x_{vw,F_1}$}%
}}}}
\end{picture}%